\documentclass[letterpaper, 10 pt, conference]{ieeeconf}
\IEEEoverridecommandlockouts
\let\labelindent\relax
\usepackage{enumitem}
\usepackage{balance}
\usepackage[font={small}]{caption}
\usepackage{subcaption}
\usepackage{array}
\usepackage{textcomp}
\usepackage{mathtools, nccmath}
\usepackage{graphicx}
\usepackage{amsfonts}
\usepackage{amsmath}
\usepackage{amssymb}
\usepackage{algorithm}
\usepackage{algorithmic}
\usepackage{hyperref}
\usepackage{tikz}
\usepackage{float}
\usepackage{arydshln}
\usepackage{multirow}
\usepackage{bm}
\usepackage{epstopdf}
\usepackage{cite}
\usepackage{siunitx}
\usepackage{enumitem}

\usepackage{amsthm}

\newtheorem{theorem}{Theorem}
\newtheorem{lemma}{Lemma}

\newtheorem{remark}{Remark}
\theoremstyle{definition}
\newtheorem{definition}{Definition}
\newtheorem{problem}{Problem}

\newtheorem{proposition}{Proposition}

\title{\LARGE \bf
    Event-Triggered Adaptive Taylor–Lagrange Control \\for Safety-Critical Systems}

\author{Shuo Liu$^{1}$, Wei Xiao$^{2}$, Christos G. Cassandras$^{3}$ and Calin A. Belta$^{4}$
\thanks{This work was supported in part by the NSF under grant IIS-2024606 at Boston University and by a Brendan Iribe endowed professorship at the University of Maryland.}
\thanks{$^{1}$S. Liu is with the Department of Mechanical Engineering, Boston
University, Brookline, MA, USA. 
        {\tt\small liushuo@bu.edu}}%
\thanks{$^{2}$W. Xiao is with the School of Electrical and Electronic Engineering, Nanyang Technological University, Singapore and MIT CSAIL, MA, USA.
        {\tt\small weixy@mit.edu}}%
\thanks{$^{3}$C.G. Cassandras is with the Division of Systems Engineering, Boston
University, USA.
        {\tt\small cgc@bu.edu}}%
\thanks{$^{4}$C. Belta is with the Department of Electrical and Computer Engineering and the Department of Computer Science, University of Maryland, College Park, MD, USA. 
        {\tt\small cbelta@umd.edu}}%
}

\begin{document} 
\maketitle

\begin{abstract}
This paper studies safety-critical control for nonlinear systems under sampled-data implementations of the controller.
The recently proposed Taylor--Lagrange Control (TLC) method provides safety guarantees but relies on a fixed discretization-related parameter, which can lead to infeasibility or unsafety in the presence of input constraints and inter-sampling effects. To address these limitations, we propose an adaptive Taylor--Lagrange Control (aTLC) framework with an event-triggered implementation, where the key parameter defines the discretization time interval and is selected online as state-dependent rather than fixed. This enables the controller to dynamically balance feasibility and safety by adjusting the effective time interval of the Taylor expansion. The resulting controller is implemented as a sequence of Quadratic Programs (QPs)
 with input constraints. We further introduce a selection rule to choose the largest feasible discretization time interval from a finite candidate set while preserving safety.
Simulation results on an adaptive cruise control (ACC) problem demonstrate that the proposed approach improves feasibility, guarantees safety, achieves smoother control actions, and reduces triggering events compared to TLC while avoiding the class $\mathcal K$ function tuning required by Control Barrier Function (CBF)-based methods.
\end{abstract}

\section{Introduction}
\label{sec:Introduction}
Ensuring stability while optimizing cost under safety constraints remains a fundamental challenge in autonomous systems. 
A key difficulty lies in simultaneously achieving computational tractability and rigorous safety guarantees, especially for nonlinear dynamics. 

Existing approaches can be broadly categorized into optimization-based and verification-based methods. Optimization-based approaches, such as classical optimal control and dynamic programming techniques \cite{kirk2004optimal,arthur1975applied,bellman1966dynamic,denardo2012dynamic}, 
are often tailored to linear systems or suffer from the curse of dimensionality, limiting their applicability to complex nonlinear systems. Model Predictive Control (MPC) \cite{garcia1989model,rawlings2020model} offers a practical framework for safety-critical control, but nonlinear MPC formulations are computationally intensive, and their linearized approximations may compromise safety guarantees.
Verification-based approaches, on the other hand, such as reachability-based methods \cite{aubin2011viability,mitchell2005time}, provide strong theoretical guarantees for safety, yet their computational cost remains prohibitive for real-time implementation. To bridge this gap, barrier-based methods have emerged as a promising alternative, offering a computationally efficient way to enforce safety constraints in nonlinear systems.

Barrier functions (BFs) have been widely used in optimization to handle inequality constraints, for instance by incorporating reciprocal barrier terms into the cost function \cite{boyd2004convex}. They have also been adopted in learning-based frameworks, such as safe Reinforcement Learning (RL) \cite{cheng2019end}, to encourage safety during training. 
However, in these formulations, safety is typically encoded as part of the cost or reward, resulting in \textit{soft} constraint enforcement rather than strict safety guarantees, i.e., ensuring that safety constraints are satisfied at all times.
Alternatively, barrier functions have been employed as Lyapunov-like certificates \cite{wieland2007constructive} for system verification and control \cite{tee2009barrier,prajna2007framework,wisniewski2015converse}, enabling the characterization of safe invariant sets. However, these methods typically focus on safety verification rather than control synthesis, which limits their direct applicability in real-time control design under input constraints.

Control BFs (CBFs) extend barrier functions by explicitly incorporating control inputs to enforce forward invariance of safe sets for affine control systems. If a CBF satisfies certain Lyapunov-like conditions, safety can be guaranteed in the sense of set forward invariance~\cite{ames2016control}. By combining CBFs with Control Lyapunov Functions (CLFs), the CBF-CLF-QP framework formulates safety-critical control as a sequence of Quadratic Programs (QPs)~\cite{ames2012control, ames2016control}, enabling real-time implementation. Extensions of this framework have been developed to handle high-relative-degree constraints and adaptive control scenarios \cite{nguyen2016exponential, xiao2021high, xiao2021adaptive, liu2023auxiliary, liu2024auxiliary}. CBFs have also been integrated with RL to ensure safety~\cite{ahmad2025hierarchical}.
However, existing CBF-based methods exhibit several limitations. First, the use of class $\cal K$ 
functions can introduce conservativeness, since the CBF condition constitutes only a sufficient condition for safety and may overly restrict the set of admissible control inputs. Second, these methods require the selection of class $\cal K$ functions, which introduces additional parameters that are often difficult to tune in practice. This challenge is exacerbated in high-order CBF formulations \cite{nguyen2016exponential, xiao2021high}, where multiple class $\cal K$ functions must be specified, further increasing the tuning burden. Moreover, under sampled-data implementation, CBF-based controllers may suffer from inter-sampling issue, i.e., potential violations of safety constraints between discrete update instants when control inputs are held constant \cite{liu2026sampling}.

The recently proposed Taylor–Lagrange Control (TLC) method~\cite{xiao2025taylor} ensures system safety by leveraging Taylor’s theorem with Lagrange remainder~\cite{taylor1717methodus,de1813theorie}. Compared with CBF-based approaches, TLC can reduce the need for class $\mathcal K$ function tuning. Moreover, similar to CBF-based methods, TLC leads to a QP formulation, enabling efficient real-time implementation.
To address the inter-sampling issue, a robust variant of TLC (rTLC) \cite{xiao2026robust} has been proposed to guarantee safety over the entire inter-event interval. 
However, existing TLC-based approaches rely on manually tuned parameters that are kept constant over time. This can lead to infeasibility of the resulting QP, particularly in the presence of tight control bounds, due to conflicts between the TLC safety constraints and input constraints. To address this issue, we propose Adaptive TLC (aTLC) to safety-critical control problems. Specifically, the contributions are as follows:
\begin{itemize}
    \item An adaptive Taylor--Lagrange Control (aTLC) framework that defines the discretization time interval as a state-dependent variable selected online, enabling improved feasibility.
    
    \item An event-triggered implementation of aTLC, where control updates are performed only when the system state exits a prescribed neighborhood or when the selected time interval elapses, mitigating inter-sampling effects while maintaining safety guarantees.
    
    \item A value-function-based characterization of feasibility via a \emph{margin} function and the properties of the maximum feasible discretization time interval. Based on this insight, we develop an adaptive selection rule that chooses the largest feasible discretization time interval from a finite candidate set while maintaining safety.

    \item Simulation results on an adaptive cruise control (ACC) problem. We demonstrate that the proposed method achieves improved feasibility, guaranteed safety, smoother control actions, and fewer triggering events compared to non-adaptive TLC.
\end{itemize}

\section{Definitions and Preliminaries}
\label{sec:Preliminaries}

Consider an affine control system of the form
\begin{equation}
\label{eq:affine-control-system}
\dot{\boldsymbol{x}}=f(\boldsymbol{x})+g(\boldsymbol{x})\boldsymbol{u},
\end{equation}
 where $\boldsymbol{x}\in \mathbb{R}^{n}, f:\mathbb{R}^{n}\to\mathbb{R}^{n}$ and $g:\mathbb{R}^{n}\to\mathbb{R}^{n\times q}$ are locally Lipschitz, and $\boldsymbol{u}\in \mathcal U\subset \mathbb{R}^{q}$, where $\mathcal U$ denotes the control limitation set, which is assumed to be in the form: 
\begin{equation}
\label{eq:control-constraint}
\mathcal U \coloneqq \{\boldsymbol{u}\in \mathbb{R}^{q}:\boldsymbol{u}_{min}\le \boldsymbol{u} \le \boldsymbol{u}_{max} \}, 
\end{equation}
with $\boldsymbol{u}_{min},\boldsymbol{u}_{max}\in \mathbb{R}^{q}$ (vector inequalities are interpreted componentwise). We assume that no component of $\boldsymbol{u}_{min}$ and $\boldsymbol{u}_{max}$ can be infinite. 

\begin{definition}[Class $\cal K$ function~\cite{Khalil:1173048}]
\label{def:class-k-f}
A continuous function $\alpha:[0,a)\to[0,+\infty),a>0$ is called a class $\cal K$ function if it is strictly increasing and $\alpha(0)=0.$
\end{definition}

\begin{definition}
\label{def:forward-inv}
A set $\mathcal C\subset \mathbb{R}^{n}$ is forward invariant for system \eqref{eq:affine-control-system} if its solutions for some $\boldsymbol{u} \in \mathcal U$ starting from any $\boldsymbol{x}(0) \in \mathcal C$ satisfy $\boldsymbol{x}(t) \in \mathcal C, \forall t \ge 0.$
\end{definition}

\begin{definition}
\label{def:relative-degree}
The relative degree of a differentiable function $h:\mathbb{R}^{n}\to\mathbb{R}$ is the minimum number of times we need to differentiate it along dynamics \eqref{eq:affine-control-system} until any component of $\boldsymbol{u}$ explicitly shows in the corresponding derivative. 
\end{definition}

In this paper, the \textbf{safety requirement} is defined by the constraint $h(\boldsymbol{x}) \ge 0$, and \textbf{safety} refers to the forward invariance of the set
\begin{equation}
\label{eq: safety set}
\mathcal{C} \coloneqq \{\boldsymbol{x} \in \mathbb{R}^n : h(\boldsymbol{x}) \ge 0\}.
\end{equation}
The relative degree of $h$ is thus referred to as the relative degree of the safety requirement.


\begin{definition}[Taylor--Lagrange Control (TLC)~\cite{xiao2025taylor}]
\label{def: TLC}
A continuously differentiable function $h:\mathbb{R}^n \to \mathbb{R}$ 
is called a Taylor--Lagrange Control (TLC) function of relative degree $m$ 
for system~\eqref{eq:affine-control-system} if
\begin{equation}
\label{eq: TLC condition}
\begin{aligned}
\sup_{\boldsymbol{u}(\xi)\in \mathcal{U}} \Bigg[
& \sum_{k=0}^{m-1} \frac{L_f^k h(\boldsymbol{x}(t_0))}{k!}(t-t_0)^k + \frac{L_f^m h(\boldsymbol{x}(\xi))}{m!}(t-t_0)^m \\
& + \frac{L_g L_f^{m-1} h(\boldsymbol{x}(\xi))\,\boldsymbol{u}(\xi)}{m!}(t-t_0)^m
\Bigg] \ge 0,
\end{aligned}
\end{equation}
for all $\boldsymbol{x}(t_0)\in \mathcal{C}$, $t_0\ge0$, and $t>t_0$, where $\xi \in (t_0,t)$ is an intermediate point provided by Taylor's theorem with Lagrange remainder along the corresponding system trajectory.
Here, $L_f h$ and $L_g h$ denote the Lie derivatives of $h$ along $f$ and $g$, respectively.
\end{definition}

\begin{theorem}[~\cite{xiao2025taylor}]
\label{thm：thm1}
Let $h(\boldsymbol{x})$ be a TLC function as defined in Def.~\ref{def: TLC}, and let the corresponding safe set $\mathcal{C}$ be defined as in~\eqref{eq: safety set}.
If $h(\boldsymbol{x}(t_0)) \ge 0$, then any Lipschitz continuous control input
$\boldsymbol{u}(\cdot)$ that makes the expression inside the supremum in \eqref{eq: TLC condition} nonnegative for every $t>t_0$ with $\xi \in (t_0,t)$ provided by Taylor's theorem with Lagrange remainder renders the set
$\mathcal{C}$ forward invariant for system~\eqref{eq:affine-control-system}.
\end{theorem}

It follows from Taylor's theorem with Lagrange remainder \cite{taylor1717methodus,de1813theorie} that along a given controlled trajectory and at the corresponding intermediate point $\xi$, the expression inside the supremum in \eqref{eq: TLC condition} equals $h(\boldsymbol{x}(t))$. Therefore, the nonnegativity of this expression is \textit{necessary and sufficient} for $h(\boldsymbol{x}(t))\ge0$.

In contrast, the $m$-th order condition in High-Order Control Barrier Functions (HOCBFs) [\text{Eq}. (13), \cite{xiao2021high}] involves $m$ class $\cal K$ functions, introducing additional design degrees of freedom and rendering the condition sufficient, \emph{but not necessary}, for safety. Moreover, these functions require tuning multiple parameters in practice. On the other hand, the $m$-th order TLC condition in \eqref{eq: TLC condition} involves a single design parameter, namely the time interval $(t-t_0)$, without requiring multiple tuning parameters.

\begin{definition}[CLF~\cite{ames2012control}]
\label{def:control-l-f}
A continuously differentiable function $V:\mathbb{R}^{n}\to\mathbb{R}$ is an exponentially stabilizing Control Lyapunov Function (CLF) for system \eqref{eq:affine-control-system} if there exist constants $c_{1}>0, c_{2}>0,c_{3}>0$ such that for $\forall \boldsymbol{x} \in \mathbb{R}^{n}, c_{1}\left \|  \boldsymbol{x} \right \| ^{2} \le V(\boldsymbol{x}) \le c_{2}\left \|  \boldsymbol{x} \right \| ^{2}$ and
\begin{equation}
\label{eq:clf}
\inf_{\boldsymbol{u}\in \mathcal U}[L_{f}V(\boldsymbol{x})+L_{g}V(\boldsymbol{x})\boldsymbol{u}+c_{3}V(\boldsymbol{x})]\le 0.
\end{equation}
\end{definition}

Several works (e.g., \cite{nguyen2016exponential, xiao2021high}) address safety-critical control by integrating HOCBFs with quadratic cost objectives, resulting in Optimal Control Problems (OCPs) for systems with high relative degree. 
In practice, these OCPs are implemented in real time through a sequence of QPs. In these frameworks, HOCBF constraints ensure forward invariance of the safe set, while CLFs~\eqref{eq:clf} can be incorporated as soft constraints to enforce exponential convergence to desired states \cite{xiao2021high}.
Similarly, the TLC condition \eqref{eq: TLC condition}
can be employed to enforce safety within a QP framework. By combining TLC-based safety constraints with CLF-based objectives, one can simultaneously guarantee safety and achieve exponential convergence of the desired states.

\section{Problem Formulation and Approach}
\label{sec:Problem Formulation and Approach}
Our goal is to generate a control strategy for system \eqref{eq:affine-control-system} 
that ensures convergence of the system state to a desired equilibrium, 
minimizes control effort, satisfies safety requirements, and respects input constraints.

\textbf{Objective:} We consider the cost  
\begin{equation}
\label{eq:cost-function-1}
\begin{split}
 J(\boldsymbol{u}(t))=\int_{0}^{T} 
 \| \boldsymbol{u}(t) \| ^{2}dt+p\left \| \boldsymbol{x}(T)-\boldsymbol{x}_{e} \right \| ^{2},
\end{split}
\end{equation}
where $\|\cdot\|$ denotes the Euclidean norm, $T>0$ is the terminal time, $p>0$ is a weighting factor, and $\boldsymbol{x}_e \in \mathbb{R}^n$ is the desired equilibrium state of system~\eqref{eq:affine-control-system}. The cost term $p\|\boldsymbol{x}(T)-\boldsymbol{x}_e\|^2$ promotes convergence of the state to $\boldsymbol{x}_e$.

\textbf{Safety Requirement:} System \eqref{eq:affine-control-system} should always satisfy one or more safety requirements of the form: 
\begin{equation}
\label{eq:Safety constraint}
h(\boldsymbol{x})\ge 0, \boldsymbol{x} \in \mathbb{R}^{n}, \forall t \in [0, T],
\end{equation}
where $h:\mathbb{R}^{n}\to\mathbb{R}$ is assumed to be a continuously differentiable equation. 

\textbf{Control Limitations:} The controller $\boldsymbol{u}$ should always satisfy \eqref{eq:control-constraint} for all $t \in [0, T].$

A control policy is \textbf{feasible} if \eqref{eq:Safety constraint} and \eqref{eq:control-constraint} are satisfied $\forall t \in [0, T].$ In this paper, we consider the following problem:

\begin{problem}
\label{prob:smooth-prob}
Find a feasible control policy for system \eqref{eq:affine-control-system} such that cost \eqref{eq:cost-function-1} is minimized.
\end{problem}

Existing TLC \cite{xiao2025taylor} and its robust variant (rTLC) \cite{xiao2026robust} are both implemented by solving a QP at discrete update instants in either time-driven or event-driven manner. 
While event-triggered TLC \cite{xiao2025taylor} provides safety guarantees and rTLC mitigates inter-sampling effects, both require a fixed time interval $(t-t_0)$ in \eqref{eq: TLC condition}, which is manually selected in advance and remains unchanged at every QP update.
 Such a fixed choice cannot adapt to the current state or the available control authority. If chosen too aggressively, the resulting TLC/rTLC safety condition may become overly restrictive and conflict with input bounds, leading to infeasibility of one or more QPs; if chosen too conservatively, it may lead to unnecessarily frequent controller updates and increased computational burden. Hence, a constant time interval can degrade both optimization feasibility and computational efficiency, motivating the need for an adaptive TLC approach.

\textbf{Approach:} 
To solve Problem \ref{prob:smooth-prob} and address the limitations of non-adaptive TLC, we introduce an adaptive framework in which the time interval $(t-t_0)$ is parameterized by $\tau$, which is selected online based on the current state. To enforce convergence to the desired state, we select a CLF and impose the corresponding CLF constraint \eqref{eq:clf}. To satisfy the safety requirement, we construct an aTLC function and convert it into the corresponding aTLC constraint. The CLF and aTLC constraints are jointly imposed in a QP with input bounds. 
Within the aTLC formulation, the parameter $\tau$ is selected from a finite candidate set at each update instant, improving feasibility under input constraints. This adaptive scheme is combined with an event-triggered implementation, where control updates are executed when the state exits a prescribed neighborhood or when the selected time interval elapses, thereby ensuring safety over inter-sampling intervals.

\section{Adaptive Taylor–Lagrange Control}
\label{sec:ATLC}
In this section, we develop an adaptive TLC (aTLC) framework. 
The key idea is to treat the time interval appearing in the Taylor--Lagrange expansion 
as a state-dependent parameter that can be adjusted online to improve 
feasibility and mitigate inter-sampling effects.

\begin{definition}[Adaptive Taylor--Lagrange Control (aTLC)]
\label{def:ATLC}
Consider system \eqref{eq:affine-control-system} with safety requirement $h(\boldsymbol{x})\ge 0$, where
$h$ has relative degree $m$. For any state $\boldsymbol{x}(t_0)$ and any parameter $\tau \in [\tau_{\min},\tau_{\max}]$ with predefined bounds, 
define the $\tau$-parameterized adaptive TLC
condition as
\begin{equation}
\begin{aligned}
&\sup_{\boldsymbol{u}(\xi)\in U}
\Bigg[
\sum_{i=0}^{m-1}\frac{L_f^i h(\boldsymbol{x}(t_0))}{i!}\sigma^i +\frac{L_f^m h(\boldsymbol{x}(\xi))}{m!}\sigma^m+\\
&+
\frac{L_gL_f^{m-1}h(\boldsymbol{x}(\xi))\boldsymbol{u}(\xi)}{m!}\sigma^m
\Bigg]
\ge 0,~ \forall \sigma\in(0,\tau]
\end{aligned},
\label{eq:tau_tlc_condition}
\end{equation}
where for each $\sigma \in (0,\tau]$,  $\xi\in(t_0,t_0+\sigma)$ is the intermediate point given by Taylor's
theorem with Lagrange remainder.
An \emph{Adaptive Taylor--Lagrange Control} (aTLC) function is a function 
$h$ for which the parameter $\tau$ is not fixed a priori, 
but selected online as a state-dependent variable
\begin{equation}
\label{eq: practical tau}
\tau = \mu(\boldsymbol{x}(t_0)),
\end{equation}
where $\mu(\cdot)$ denotes a state-dependent policy, which may be defined explicitly or implicitly. 
The corresponding control input 
is then computed using the aTLC condition associated with $\tau$ selected from \eqref{eq: practical tau}.
\end{definition}

\begin{theorem}
\label{thm: aTLC safety}
Consider system \eqref{eq:affine-control-system} with safe set
$\mathcal{C}=\{\boldsymbol{x}:h(\boldsymbol{x})\ge0\}$.
If $h(\boldsymbol{x}(t_0))\ge0$, then any admissible Lipschitz continuous control
$\boldsymbol{u}(\cdot)$ making the expression inside the supremum in
\eqref{eq:tau_tlc_condition} nonnegative for all
$\sigma\in(0,\tau]$ at every $t_0\ge0$, where
$\tau\in[\tau_{\min},\tau_{\max}]$ is selected by
\eqref{eq: practical tau} and $\xi\in(t_0,t_0+\sigma)$,
renders $\mathcal{C}$ forward invariant.
\end{theorem}

\begin{proof}
The proof follows directly from Taylor's theorem with Lagrange remainder and
the proof of the original TLC result \cite{xiao2025taylor}. For any $t_0\ge0$
with $h(\boldsymbol{x}(t_0))\ge0$ and any
$\tau\in[\tau_{\min},\tau_{\max}]$, for each
$\sigma\in(0,\tau]$ there exists an intermediate point
$\xi\in(t_0,t_0+\sigma)$ such that
\begin{equation}
\begin{aligned}
h(\boldsymbol{x}(t_0&+\sigma))
=
\sum_{i=0}^{m-1}\frac{L_f^i h(\boldsymbol{x}(t_0))}{i!}\sigma^i \\
&+\frac{L_f^m h(\boldsymbol{x}(\xi))
+L_gL_f^{m-1}h(\boldsymbol{x}(\xi))\boldsymbol{u}(\xi)}{m!}\sigma^m .
\end{aligned}
\end{equation}
Hence, if a Lipschitz continuous control input $\boldsymbol{u}(\xi)$
satisfies \eqref{eq:tau_tlc_condition} for all $\sigma\in(0,\tau]$, then
$h(\boldsymbol{x}(t_0+\sigma)) \ge 0$ for all $\sigma\in(0,\tau]$. Thus, the state remains in $\mathcal C$ over
$[t_0,t_0+\tau]$. Since this argument can be recursively applied from any
$t_0$ with $h(\boldsymbol{x}(t_0))\ge0$, we obtain
$h(\boldsymbol{x}(t))\ge0$ for all $t\ge t_0$.
Therefore, $\mathcal C$ is forward invariant. This applies in particular to
the state-dependent selection
$\tau=\mu(\boldsymbol{x}(t_0))$ in Def.~\ref{def:ATLC}.
\end{proof}

Unlike Theorem~\ref{thm：thm1}, Theorem~\ref{thm: aTLC safety} allows the parameter $\tau$ to be selected online as a state-dependent variable rather than fixed a priori. Although the aTLC condition \eqref{eq:tau_tlc_condition} is exact, its implementation is complicated by the unknown intermediate point $\xi \in (t_0,t_0+\sigma)$. Therefore, in implementation one can only construct an approximate aTLC condition using the information available at time $t_0$ or over a local neighborhood of $\boldsymbol{x}(t_0)$. The resulting approximation of $h(\boldsymbol{x}(t_0+\sigma))$ generally differs from its exact value, and the discrepancy depends on both the current state $\boldsymbol{x}(t_0)$ and the selected parameter $\tau$. Although this discrepancy decreases as $\tau$ approaches zero, it does not vanish completely for nonzero $\tau$. Consequently, to improve implementation-level safety in the presence of such inter-sampling errors, we adopt an event-triggered framework that updates the control input and reconstructs the aTLC condition whenever the state exits a prescribed neighborhood or when the selected time
interval elapses.
rTLC \cite{xiao2026robust} addresses inter-sampling effects, but relies on a fixed time interval, 
which can lead to infeasibility under tight control bounds. 

\subsection{Event-Triggered aTLC with Adaptive Time Interval}
\label{subsec: Event triggered aTLC}
We consider the event-triggered implementation of TLC as in \cite{xiao2025taylor}, 
and extend it by introducing an adaptive time interval. Let $\{t_k\}_{k\ge0}$ denote the sequence of event times defined by
\begin{equation}
\label{eq:event time}
t_{k+1}
=
\min\left\{
\inf \left\{
t > t_k :
\boldsymbol{x}(t)\notin S(\boldsymbol{x}_k)
\right\},
\;
t_k+\tau_k
\right\},
\end{equation}
where the neighborhood $S(\boldsymbol{x}_k)$ is defined as a hyper-rectangle of the form
\begin{equation}
\label{eq: state bound}
S(\boldsymbol{x}_k) := \{ \boldsymbol{x} : \boldsymbol{x}_k - \underline{\boldsymbol{x}} \le \boldsymbol{x} \le \boldsymbol{x}_k + \overline{\boldsymbol{x}} \},
\end{equation}
where $\underline{\boldsymbol{x}}, \overline{\boldsymbol{x}} \in \mathbb{R}^n_{>0}$ are given vectors that define the size of the neighborhood in each state dimension. At each event time $t_k$, we set $\boldsymbol{x}_k := \boldsymbol{x}(t_k)$. 

For a given $\tau$, we define the robust aTLC-related quantities
\begin{align}
\label{eq: ra}
h_{\mathrm{ra}}(\boldsymbol{x}_k, \tau) 
&:= 
\min_{\substack{
\boldsymbol{x} \in S(\boldsymbol{x}_k)\\
\sigma \in [0,\tau]
}}
\Big(
\sum_{i=0}^{m-1} 
\frac{L_f^i h(\boldsymbol{x}_k)}{i!} \sigma^i
+
\frac{L_f^m h(\boldsymbol{x})}{m!} \sigma^m
\Big),
\end{align}
where for each component $j \in \{1,\dots,q\}$ we define
\begin{equation}
\label{eq: Gra}
\begin{aligned}
\underline{G}_{\mathrm{ra},j}(\boldsymbol{x}_k,\tau)
&:=\min_{\substack{\boldsymbol{x}\in S(\boldsymbol{x}_k)\\
\sigma\in[0,\tau]}}
\frac{\sigma^m}{m!}\big[\phi(\boldsymbol{x})\big]_j,\\
\overline{G}_{\mathrm{ra},j}(\boldsymbol{x}_k,\tau)
&:=\max_{\substack{\boldsymbol{x}\in S(\boldsymbol{x}_k)\\
\sigma\in[0,\tau]}}
\frac{\sigma^m}{m!}\big[\phi(\boldsymbol{x})\big]_j,
\end{aligned}
\end{equation}
where $\phi(\boldsymbol{x}) := L_g L_f^{m-1} h(\boldsymbol{x}), \boldsymbol{u} = (u_1,\dots,u_q)$. Introducing an auxiliary variable $\boldsymbol{r} = (r_1,\dots,r_q)$ satisfying $r_j\le
\underline{G}_{\mathrm{ra},j}u_j,
r_j\le
\overline{G}_{\mathrm{ra},j}u_j$
ensures that $r_j$ lower bounds the worst-case contribution of the
$j$-th control component without requiring the sign of $u_j$ to be
known a priori. Consequently,
$\sum_{j=1}^{q}r_j+h_{\mathrm{ra}}(\boldsymbol{x}_k,\tau)$
provides a valid lower bound on the robust aTLC expression. The event-triggered aTLC condition becomes
\begin{equation}
\label{eq: aTLC condition}
\sup_{\substack{
r_j\le\underline G_{\mathrm{ra},j}u_j\\
r_j\le\overline G_{\mathrm{ra},j}u_j,~\boldsymbol u\in\mathcal U
}}
\left[
\sum_{j=1}^{q}r_j+
h_{\mathrm{ra}}(\boldsymbol{x}_k,\tau)
\right]\ge0.
\end{equation}

At each event time $t_k$, the current state
$\boldsymbol{x}_k := \boldsymbol{x}(t_k)$ is measured, a local set
$S(\boldsymbol{x}_k)$ is constructed, and a parameter
$\tau_k \in [\tau_{\min}, \tau_{\max}]$ is selected as detailed in
Sec.~\ref{subsec: tau choice}. The corresponding robust bounds
$h_{\mathrm{ra}}$, $\underline{G}_{\mathrm{ra},j}$, and
$\overline{G}_{\mathrm{ra},j}$ are then computed, and the control input is
obtained by solving the following QP:
\begin{equation}
\label{eq: QPs}
\begin{aligned}
&\min_{\boldsymbol{u},\delta,\boldsymbol{r}}
\quad \|\boldsymbol{u}\|^2+w\delta^2 \\
\text{s.t.}\quad
&\sum_{j=1}^{q} r_j
+h_{\mathrm{ra}}(\boldsymbol{x}_k,\tau_k)\ge0,\\
&r_j\le
\underline{G}_{\mathrm{ra},j}(\boldsymbol{x}_k,\tau_k)u_j, ~
r_j\le
\overline{G}_{\mathrm{ra},j}(\boldsymbol{x}_k,\tau_k)u_j,\\
&L_fV(\boldsymbol{x}_k)+L_gV(\boldsymbol{x}_k)\boldsymbol{u}
+c_3V(\boldsymbol{x}_k)\le\delta, ~ \boldsymbol{u}\in\mathcal U,\\
\end{aligned}
\end{equation}
where $\delta\ge0$ is a slack variable that relaxes the CLF constraint \eqref{eq:clf}.
The resulting control input $\boldsymbol{u}_k$ is applied as 
$\boldsymbol{u}(t)=\boldsymbol{u}_k$ for $t \in [t_k, t_{k+1})$. 
The event time is then updated to $t_{k+1}$ defined in \eqref{eq:event time}, and the procedure is repeated 
until the final time $T$. Importantly, $\tau_k$ is not necessarily equal to the inter-event interval 
$(t_{k+1}-t_k)$. It is selected at $t_k$ as the maximum valid horizon 
of the current aTLC condition and therefore directly imposes the bound
$t_{k+1}-t_k \leq \tau_k$ through the event-triggering rule 
\eqref{eq:event time}. Meanwhile, $\tau_k$ also affects the control input 
through \eqref{eq: QPs}, which influences the state evolution and may cause 
an earlier state-triggered update when $\boldsymbol{x}(t)$ exits 
$S(\boldsymbol{x}_k)$.

\subsection{Forward Invariance with Adaptive Time Interval}

We first show that introducing an adaptive time interval does not affect the continuous-time safety guarantee.

\begin{theorem}[Forward Invariance under Event-Triggered aTLC]
\label{thm: thm3}
Consider system \eqref{eq:affine-control-system} with safe set \eqref{eq: safety set}. Assume $\boldsymbol{x}(0)\in\mathcal C$, the closed-loop
trajectory exists for all $t\ge0$, and $t_k\to\infty$ as $k\to\infty$.
Suppose that $h$ is an aTLC function of relative degree $m$ defined in Def. \ref{def:ATLC}, and that the control is implemented under the event-triggered aTLC framework described in Sec. \ref{subsec: Event triggered aTLC}. Let the event times \(\{t_k\}_{k\ge 0}\) be generated by \eqref{eq:event time} and let \(\tau_k\in[\tau_{\min},\tau_{\max}]\) be selected by
\eqref{eq: practical tau} at each event time \(t_k\) and held constant over \([t_k,t_{k+1})\). If at every event time $t_k$ there exist
$\boldsymbol{u}_k\in\mathcal U$ and auxiliary variables
$\boldsymbol{r}_k$ satisfying $r_{k,j}\le
\underline{G}_{\mathrm{ra},j}u_{k,j},
r_{k,j}\le
\overline{G}_{\mathrm{ra},j}u_{k,j}$ such that
\begin{equation}
\label{eq: aTLC constraint}
\sum_{j=1}^{q} r_{k,j}
+h_{\mathrm{ra}}(\boldsymbol{x}_k,\tau_k)\ge0,
\end{equation}
then $\mathcal C$ is forward invariant for system
\eqref{eq:affine-control-system}. 
\end{theorem}

\begin{proof}
Fix any event interval $[t_k,t_{k+1})$. By \eqref{eq:event time},
$\boldsymbol{x}(t)\in S(\boldsymbol{x}_k)$ and
$\sigma:=t-t_k\in[0,\tau_k]$ for all $t\in[t_k,t_{k+1})$.
Since $\tau_k$ and $\boldsymbol{u}_k$ are held constant over this interval,
the definitions of $h_{\mathrm{ra}}$,
$\underline{G}_{\mathrm{ra},j}$, and
$\overline{G}_{\mathrm{ra},j}$ imply that the aTLC expression is lower
bounded by
$\sum_{j=1}^{q} r_{k,j}
+h_{\mathrm{ra}}(\boldsymbol{x}_k,\tau_k)$.
By \eqref{eq: aTLC constraint}, this lower bound is nonnegative; hence
\eqref{eq:tau_tlc_condition} holds for all
$t\in[t_k,t_{k+1})$. Since the event intervals cover
$[0,\infty)$, we conclude that
$h(\boldsymbol{x}(t))\ge0$ for all $t\ge0$. Hence, $\mathcal C$ is forward invariant.
\end{proof}

\subsection{Feasibility Characterization via Value Function}
\label{subsec: Value Function}
To characterize feasibility, based on Eqs. \eqref{eq: ra}--\eqref{eq: aTLC condition}, we define the set of admissible controls
\begin{equation}
\label{eq:robust aTLC cons}
U(\boldsymbol{x},\tau)
:=
\{
\boldsymbol{u}\in\mathcal U :
\exists\,\boldsymbol{r}\ \text{s.t.}\
\sum_{j=1}^{q} r_j+h_{\mathrm{ra}}(\boldsymbol{x},\tau)\ge0\}, 
\end{equation}
where $\boldsymbol{r}$ is subject to
the constraints in Sec.~\ref{subsec: Event triggered aTLC}.
\begin{proposition}[Nested Feasible Sets]
\label{prop:nested feasibility}
For any $\tau_1,\tau_2\in[\tau_{\min},\tau_{\max}]$ with
$\tau_1\leq\tau_2$, the admissible control sets satisfy
$U(\boldsymbol{x},\tau_2)\subseteq U(\boldsymbol{x},\tau_1)$.
\end{proposition}
\begin{proof}
Since $[0,\tau_1]\subseteq[0,\tau_2]$, the robust aTLC bounds for
$\tau_2$ are evaluated over a superset of those for $\tau_1$.
Hence, any control input satisfying the robust aTLC constraint \eqref{eq:robust aTLC cons} for
$\tau_2$ also satisfies it for $\tau_1$.
\end{proof}
We then define the maximum feasible parameter
\begin{equation}
\tau^*(\boldsymbol{x}) := \sup \{ \tau \in [\tau_{\min},\tau_{\max}] : U(\boldsymbol{x}, \tau) \neq \emptyset \}.
\end{equation}
This definition converts the time interval into a value function that explicitly captures feasibility.
The dependence of feasibility on the parameter $\tau$ is central to the proposed aTLC design.
Since the robust aTLC condition is constructed over the entire interval $\sigma\in[0,\tau]$,
a larger $\tau$ enlarges the interval over which the worst-case bounds are evaluated and therefore
cannot make the safety condition less restrictive. Consequently, feasibility at $\tau$ implies feasibility for any smaller time interval, as established in Proposition~\ref{prop:nested feasibility}. To characterize this feasibility property, define the
robust aTLC \emph{margin}:
\begin{equation}
M(\boldsymbol{x},\tau)
:=
\sup_{\substack{
r_j\le\underline G_{\mathrm{ra},j}u_j\\
r_j\le\overline G_{\mathrm{ra},j}u_j,~\boldsymbol u\in\mathcal U
}}
\left(
\sum_{j=1}^{q}r_j
+h_{\mathrm{ra}}(\boldsymbol{x},\tau)
\right),
\end{equation}
By definition, feasibility is equivalent to
\begin{equation}
U(\boldsymbol{x}, \tau)\neq\emptyset
~ \Longleftrightarrow ~
M(\boldsymbol{x},\tau)\ge 0.
\end{equation}

\vspace*{-17pt}
\begin{algorithm}[H]
\caption{Adaptive Selection of Time Interval}
\label{alg:adaptive_tau}
\begin{algorithmic}[1]
\REQUIRE Current state $\boldsymbol{x}_k$, ordered candidate set
$\tau_1 < \tau_2 < \cdots < \tau_{N_\tau}$
\ENSURE Selected parameter $\tau_k$

\FOR{$i=N_\tau,\ldots,1$}
    \STATE Solve the QP \eqref{eq: QPs} associated with $\tau_i$
    \IF{QP is feasible}
        \STATE $\tau_k \gets \tau_i$
        \STATE \textbf{break}
    \ENDIF
\ENDFOR

\RETURN $\tau_k$
\end{algorithmic}
\end{algorithm}

\begin{lemma}
\label{lem: continuity of M}
The functions $h_{\mathrm{ra}}$,
$\underline{G}_{\mathrm{ra},j}$, and
$\overline{G}_{\mathrm{ra},j}$ are continuous in
$(\boldsymbol{x},\tau)$ for all $j$.
Consequently, $M(\boldsymbol{x},\tau)$ is also continuous.
\end{lemma}
\begin{proof}
The functions $h_{\mathrm{ra}}$,
$\underline{G}_{\mathrm{ra},j}$, and
$\overline{G}_{\mathrm{ra},j}$ are continuous by Berge's maximum
theorem \cite{berge1963topological}. Since
$r_j=\min\{\underline{G}_{\mathrm{ra},j}u_j,
\overline{G}_{\mathrm{ra},j}u_j\}$ at the optimum,
$M(\boldsymbol{x},\tau)$ is the supremum of a continuous function over
the compact set $\mathcal U$, and is therefore continuous.
\end{proof}

\begin{theorem}[Existence and Regularity of $\tau^*(\boldsymbol{x})$]
\label{thm: existence of tau}
Suppose there exists $\bar{\tau} \in [\tau_{\min},\tau_{\max}]$ such that
$U(\boldsymbol{x}, \bar{\tau}) \neq \emptyset$ for all $\boldsymbol{x}$ in a compact domain.
Then $\tau^*(\boldsymbol{x})$ is well-defined and finite. Moreover,
$\tau^*(\boldsymbol{x})$ is upper semicontinuous, i.e.,
$\limsup_{\boldsymbol{x}_{\ell} \to \boldsymbol{x}}
\tau^*(\boldsymbol{x}_{\ell}) \le \tau^*(\boldsymbol{x})$.
\end{theorem}
\begin{proof}
Based on Lemma \ref{lem: continuity of M}, $M(\boldsymbol{x},\tau)$ is continuous in $(\boldsymbol{x},\tau)$.
By assumption, $M(\boldsymbol{x},\bar{\tau})\ge 0$ for all $\boldsymbol{x}$, so the feasible set
$\{\tau\in[\tau_{\min},\tau_{\max}]:M(\boldsymbol{x},\tau)\ge 0\}$
is nonempty. As a closed subset of a compact interval, its supremum is finite and attained, hence
$\tau^*(\boldsymbol{x})$ is well-defined.
Finally, consider a sequence $\boldsymbol{x}_{\ell} \to \boldsymbol{x}$ and a subsequence such that
$\tau^*(\boldsymbol{x}_{\ell})$ converges to its largest possible limit. Since each
$\tau^*(\boldsymbol{x}_{\ell})$ is feasible at $\boldsymbol{x}_{\ell}$, continuity of the feasibility
condition implies that the limit is feasible at $\boldsymbol{x}$. By maximality of
$\tau^*(\boldsymbol{x})$, this limit cannot exceed $\tau^*(\boldsymbol{x})$. Therefore,
$\tau^*(\boldsymbol{x})$ is upper semicontinuous.
\end{proof}

Overall, Sec. \ref{subsec: Value Function} characterizes the dependence of feasibility on the parameter $\tau$ through the maximum feasible time interval $\tau^*(\boldsymbol{x})$, which serves as a bridge between the theoretical aTLC condition \eqref{eq: aTLC condition} and its state-dependent realization $\tau=\mu(\boldsymbol{x}(t_0))$ in \eqref{eq: practical tau}, enabling the online selection of the largest feasible time interval in the adaptive scheme.

\subsection{Adaptive Selection of Time Interval}
\label{subsec: tau choice}
We now propose an adaptive rule for selecting $\tau_k$.
Instead of fixing $\tau$, we select it online according to the current state.
Given a finite set of candidate values $\{\tau_i\}\subset[\tau_{\min},\tau_{\max}]$,
Alg.~\ref{alg:adaptive_tau} searches the candidates in descending order and selects the largest feasible time interval.

\begin{remark}[Recursive Feasibility]
The safety guarantees in Theorems~\ref{thm：thm1}--\ref{thm: thm3} are conditional on the existence of an admissible control input satisfying the corresponding TLC/aTLC condition. Such an input is not guaranteed to exist, particularly in the presence of input bounds. The adaptive selection of $\tau$ improves feasibility by searching for a feasible candidate time interval at each event time, but does not guarantee that one always exists. Therefore, if at every $t_k$ there exists at least one candidate $\tau_i$ such that
$U(\boldsymbol{x}_k,\tau_i)\neq\emptyset$, Alg.~\ref{alg:adaptive_tau} maintains recursive feasibility. Establishing conditions that guarantee the existence of such a control input for all future times remains an important direction for future work.
\end{remark}

In sampled-data implementations of HOCBFs, the sampling interval $\Delta t$ must be selected a priori, and an inappropriate choice may lead to inter-sampling safety violations \cite{liu2026sampling}. In contrast, aTLC explicitly accounts for inter-sampling behavior through the robust aTLC condition and the event-triggering mechanism, while avoiding the tuning of multiple class $\mathcal{K}$ functions. The bounds $\tau_{\min}$ and $\tau_{\max}$ and the size of $S(\boldsymbol{x}_k)$ remain design parameters, determining the admissible certificate horizon and the region over which inter-sampling behavior is robustly bounded. The parameter $\tau$ is then selected online according to the current state and feasibility, with the largest feasible candidate chosen to extend the certificate horizon while preserving feasibility. Compared with TLC \cite{xiao2025taylor} using a fixed time interval, this adaptive selection can reduce unnecessary triggering events and avoid infeasibility caused by an overly large fixed $\tau$.

\subsection{Complexity Analysis}
At each event time $t_k$, the adaptive aTLC scheme searches the ordered
candidate set $\{\tau_i\}$ in descending order. For each candidate, the robust
bounds are evaluated and one QP is solved until the first feasible time interval
is found. Let $N_\tau$ denote the number of candidate values,
$T_{\mathrm{bound}}$ the time required to compute
$h_{\mathrm{ra}}$ and $\sum_{j=1}^{q}r_j$, and $T_{\mathrm{QP}}$ the time
required to solve one QP. The worst-case computational complexity per event is
$\mathcal{O}\big(N_\tau(T_{\mathrm{bound}}+T_{\mathrm{QP}})\big)$. In practice,
the search may terminate early once a feasible candidate is identified,
reducing the average computational cost. Moreover, since the bound evaluations
and QPs associated with different candidate time intervals are independent, they
may be evaluated in parallel when sufficient computational resources are
available, further reducing the effective computation time.

\section{Case Study and Simulations}
\label{sec:Case Study and Simulations}
In this section, we present a case study for the use of aTLC in 
Adaptive Cruise Control (ACC) problems.
All computations are conducted in MATLAB, where the QPs are solved using 
\texttt{quadprog} and the system dynamics are integrated using 
\texttt{ode45}. The simulations are performed on an Intel\textsuperscript{\textregistered} 
Core\texttrademark{} i7-11750F CPU @ 2.50\,GHz, with an average QP computation 
time of less than 0.01\,s.

We consider nonlinear dynamics for the ego vehicle as
\begin{equation}
\label{eq:ACC-dynamics}
\underbrace{\begin{bmatrix}
\dot{z}(t) \\
\dot{v}(t) 
\end{bmatrix}}_{\dot{\boldsymbol{x}}(t)}  
=\underbrace{\begin{bmatrix}
 v_{p}-v(t) \\
 -\frac{1}{M}F_{r}(v(t))
\end{bmatrix}}_{f(\boldsymbol{x}(t))} 
+ \underbrace{\begin{bmatrix}
  0 \\
  \frac{1}{M} 
\end{bmatrix}}_{g(\boldsymbol{x}(t))}u(t),
\end{equation}
where $M$ denotes the mass of the ego vehicle, and $v_p > 0$ is the velocity of the lead vehicle. The variable $z(t)$ represents the distance between ego and the vehicle in front of it. The resistance force is modeled as $F_r(v(t)) = f_0\,\mathrm{sgn}(v(t)) + f_1 v(t) + f_2 v^2(t)$,
as in~\cite{Khalil:1173048}, where $f_0, f_1, f_2$ are positive constants determined empirically, and $v(t) > 0$ denotes the velocity of the ego vehicle. 
Vehicle limitations include constraints on safe distance, speed, and acceleration.

\textbf{Safe distance constraint:} The distance between the two vehicles is considered safe if
$z(t) \geq l_p, \forall t \in [0, T]$,
where $l_p$ denotes the minimum allowable distance.

\textbf{Speed objective:} The ego vehicle aims to achieve a desired speed $v_d > 0$.

\textbf{Acceleration constraint:} The control input $u(t)$ is constrained as $
- c_d M g \leq u(t) \leq c_a M g, \forall t \in [0, T]$,
where $g$ denotes the gravitational constant, and $c_d > 0$ and $c_a > 0$ are the deceleration and acceleration coefficients, respectively. 

The control effort is penalized by the cost functional $
\min_{u(t)} \int_0^T \left( \frac{u(t) - F_r(v(t))}{M} \right)^2 + w\delta^2dt$. The ACC problem is to find a control policy that minimizes control effort while achieving the speed objective, subject to a safe distance constraint and an acceleration constraint.  The relative degree of $z-l_{p}$ is two, and we use
a second order HOCBF, event-triggered TLC and event-triggered aTLC to
implement it by defining $h(\boldsymbol{x}) = z -l_{p} \ge 0$ and controls satisfying:

\begin{align}
\text{HOCBF:}\quad &
L_f^2 h(\boldsymbol{x}) + L_g L_f h(\boldsymbol{x})u + \nonumber \\&(p_1 + p_2)L_f h(\boldsymbol{x}) +  p_1 p_2 h(\boldsymbol{x}) \ge 0 ,
\label{eq:hocbf}\\
\text{TLC:}\quad &
L_f^2 h(\boldsymbol{x}) + L_g L_f h(\boldsymbol{x})u +\nonumber \\ &\frac{2}{\tau} L_f h(\boldsymbol{x}) + \frac{2}{\tau^2} h(\boldsymbol{x}) \ge 0 .
\label{eq:tlc}
\end{align}
If $\tau$ in \eqref{eq:tlc} is fixed ($\tau=0.5$), the method is referred to as \emph{event-triggered TLC} \cite{xiao2025taylor}. If $\tau$ is time-varying and selected according to Alg.~\ref{alg:adaptive_tau}, it is referred to as \emph{event-triggered aTLC}. For simplicity, the term “event-triggered” is omitted in the remainder of the paper.
We employ a CLF from Def. \ref{def:control-l-f} with relative degree one
to enforce the desired speed as $V(\boldsymbol{x})=(v-v_d)^2$. The parameters are $v_{p}=13.89m/s,v(0)=15 m/s, v_{d}=24m/s, M=1650kg, g=9.81m/s^{2}, z(0)=90m, l_{p}=10m, f_{0}=0.1N, f_{1}=5Ns/m, f_{2}=0.25Ns^{2}/m, c_{3}=2, w=10^{5}, \underline{\boldsymbol{x}}=\overline{\boldsymbol{x}}=0.5\cdot \mathcal{I}_{2\times1}, \tau_{\min}=0.05s,\tau_{\max}=1.6s, N_\tau=32$.

\begin{figure}[t]
    \centering    
    \vspace*{0.1cm}
    \begin{subfigure}{0.99\linewidth}
        \centering
        \includegraphics[width=\linewidth]{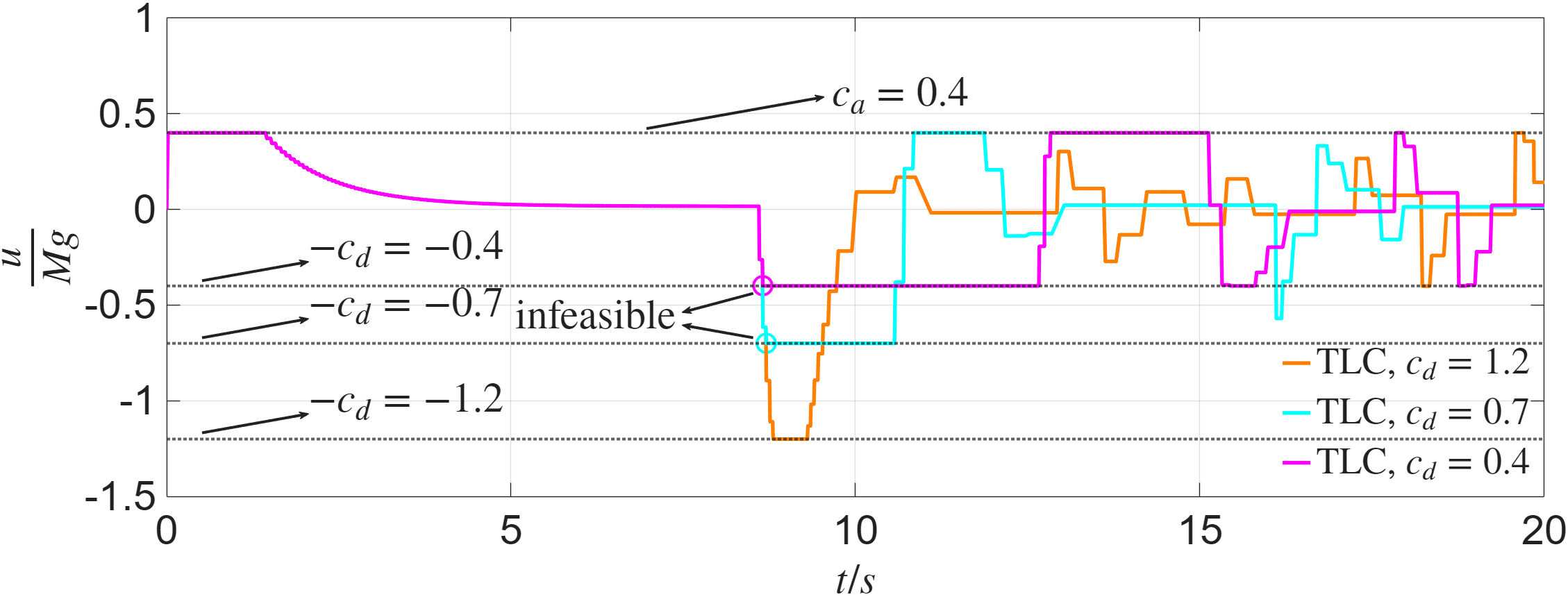}
        \caption{Control Input Profiles under Different $c_d$ (TLC)}
        \label{fig:1}
    \end{subfigure}    
    \vspace{0.2cm}    
    \begin{subfigure}{0.99\linewidth}
        \centering
        \includegraphics[width=\linewidth]{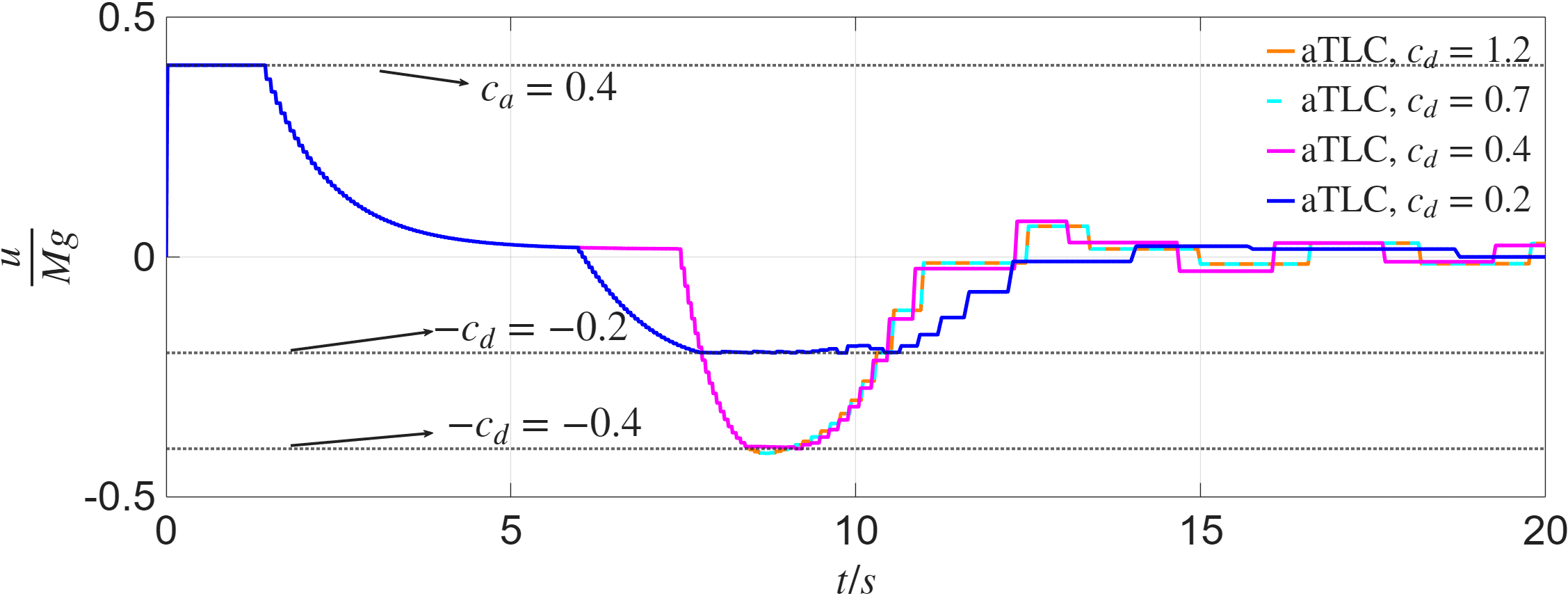}
        \caption{Control Input Profiles (orange and cyan curves overlap) under Different $c_d$ (aTLC).}
        \label{fig:2}
    \end{subfigure}    
    \vspace{0.2cm}    
    \begin{subfigure}{0.99\linewidth}
        \centering
        \includegraphics[width=\linewidth]{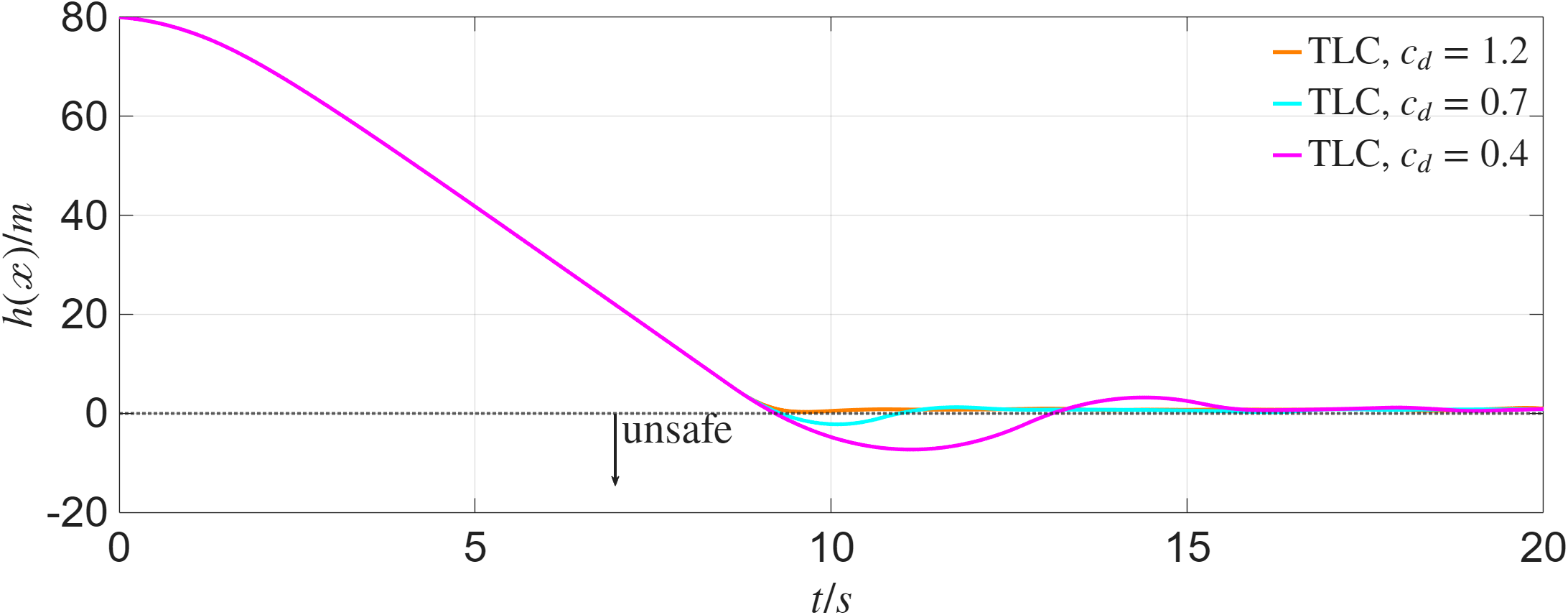}
        \caption{TLC function $h(\boldsymbol{x})$ evolution under Different $c_d$}
        \label{fig:3}
    \end{subfigure}    
    \vspace{0.2cm}    
    \begin{subfigure}{0.99\linewidth}
        \centering
        \includegraphics[width=\linewidth]{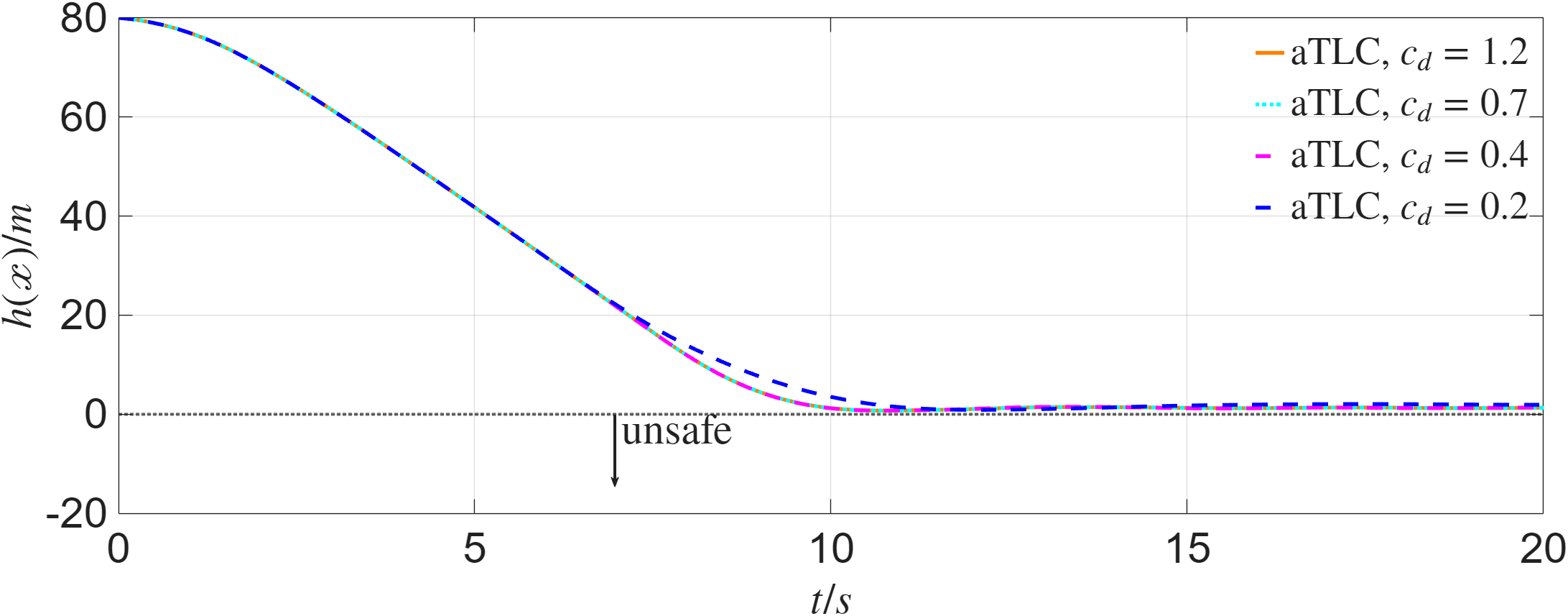}
        \caption{aTLC function $h(\boldsymbol{x})$ evolution under Different $c_d$}
        \label{fig:4}
    \end{subfigure}    
    \caption{Performance Comparison between TLC and aTLC in ACC. aTLC achieves improved feasibility while maintaining safety (i.e., avoiding violations of $h(\boldsymbol{x})\ge 0$) compared to TLC when $c_d$ is small.}
    \label{fig:4_vertical}
\end{figure}

In Fig. \ref{fig:4_vertical}, we compare the performance of TLC and aTLC under narrow control bounds. Since the ego vehicle must reach the desired speed while maintaining a safe distance from the lead vehicle, the deceleration capability is critical. A smaller $c_d$, the deceleration coefficient,
corresponds to a more slippery road condition and weaker braking capability, requiring the ego vehicle to decelerate in a timely manner to avoid safety violations.
As shown in Fig.~\ref{fig:1} and Fig.~\ref{fig:3}, TLC ensures QP feasibility and safety when $c_d = 1.2$. However, as $c_d$ decreases to $0.7$ and $0.4$, the QP becomes infeasible because no control input can satisfy both the TLC condition and the input bounds (marked by circles in the figure). In such cases, the control input is set to the maximum braking value $u = -c_d M g$ until the QP becomes feasible again. 

Fig.~\ref{fig:3} shows that, under this fallback strategy, the ego vehicle fails to maintain a safe distance, i.e., $h(\boldsymbol{x}) < 0$.
In contrast, Fig.~\ref{fig:2} and Fig.~\ref{fig:4} show that aTLC enables earlier deceleration, thereby avoiding infeasibility and safety violations. Note that Alg.~\ref{alg:adaptive_tau} selects the largest feasible time interval from the candidate set. As a result, the aTLC trajectories for $c_d=1.2$ and $0.7$ nearly coincide, whereas the trajectory for $c_d=0.4$ shows only a minor difference. Moreover, when $c_d$ is further reduced to $0.2$, increasing $\tau_{\max}$ to $3$ allows aTLC to find a feasible and safe control strategy. The input profiles also show that after $t \approx 10s$, i.e., when $h(\boldsymbol{x})$ approaches the safe-set boundary, the control input generated by aTLC varies more smoothly and stays closer to zero than that of TLC, indicating lower control effort.

In Fig.~\ref{fig:5_vertical}, we compare the performance of HOCBF, 
TLC, and aTLC under limited braking capability ($c_d = 0.4$). For HOCBF, three sets of parameters are considered, corresponding to different choices of $p_1$, $p_2$ and $\triangle t$. 
From Fig.~\ref{fig:6}, larger values of $p_1$ and $p_2$ lead to a more aggressive control strategy (i.e., delayed braking), which results in QP infeasibility around $t \approx 8\,\mathrm{s}$ (indicated by the orange circle). A larger sampling interval $\triangle t=1$ increases the inter-sampling interval and can also lead to QP infeasibility, as illustrated by the red circle. Similarly, for TLC with a fixed $\tau = 0.5$, the lack of adaptability also leads to infeasibility at approximately the same time (magenta circle). In these cases, when the QP becomes infeasible, a fallback control strategy is applied by setting the input to the maximum braking value $u = -c_d M g$ until feasibility is recovered. As shown in Fig.~\ref{fig:7}, this leads to safety violation with $h(\boldsymbol{x}) < 0$.
In contrast, reducing $p_1$ and $p_2$ or decreasing the sampling interval $\triangle t$ for HOCBFs, as well as adopting aTLC with an adaptive $\tau$, maintains QP feasibility and guarantees safety. As illustrated in Fig.~\ref{fig:5}, smaller values of $p_1$, $p_2$ and $\triangle t$ make the HOCBF controller more conservative, resulting in earlier deceleration after reaching the desired speed $v_d$ to maintain a safe distance. A similar behavior is observed for aTLC, where the vehicle gradually slows down until its speed matches that of the lead vehicle $v_p$, after which the safety distance remains nearly constant. Notably, aTLC simultaneously addresses the effects of HOCBF parameter tuning and sampling-interval selection through the adaptation of a single parameter $\tau$.
Fig.~\ref{fig:8} compares the time-varying $\tau$ in aTLC with the fixed $\tau$ in TLC, while Fig.~\ref{fig:9} shows the evolution of the event-triggered inter-event time $\Delta t_{k} = t_{k+1} - t_k$ for both methods. According to Proposition~\ref{prop:nested feasibility} and Alg.~\ref{alg:adaptive_tau}, aTLC searches $\tau$ from large to small and reduces it when a larger value becomes infeasible. This is evident for $c_d=0.4$ and $0.2$, where a smaller $\tau$ restores feasibility. As a result, after approximately $t \approx 10\,\mathrm{s}$, aTLC exhibits significantly fewer triggering events than TLC. This also leads to smoother control inputs (Fig.~\ref{fig:6}) and smoother velocity profiles (Fig.~\ref{fig:5}).

\begin{figure}[t]
    \centering    
    \vspace*{0.1cm}
    \begin{subfigure}{0.99\linewidth}
        \centering
        \includegraphics[width=\linewidth]{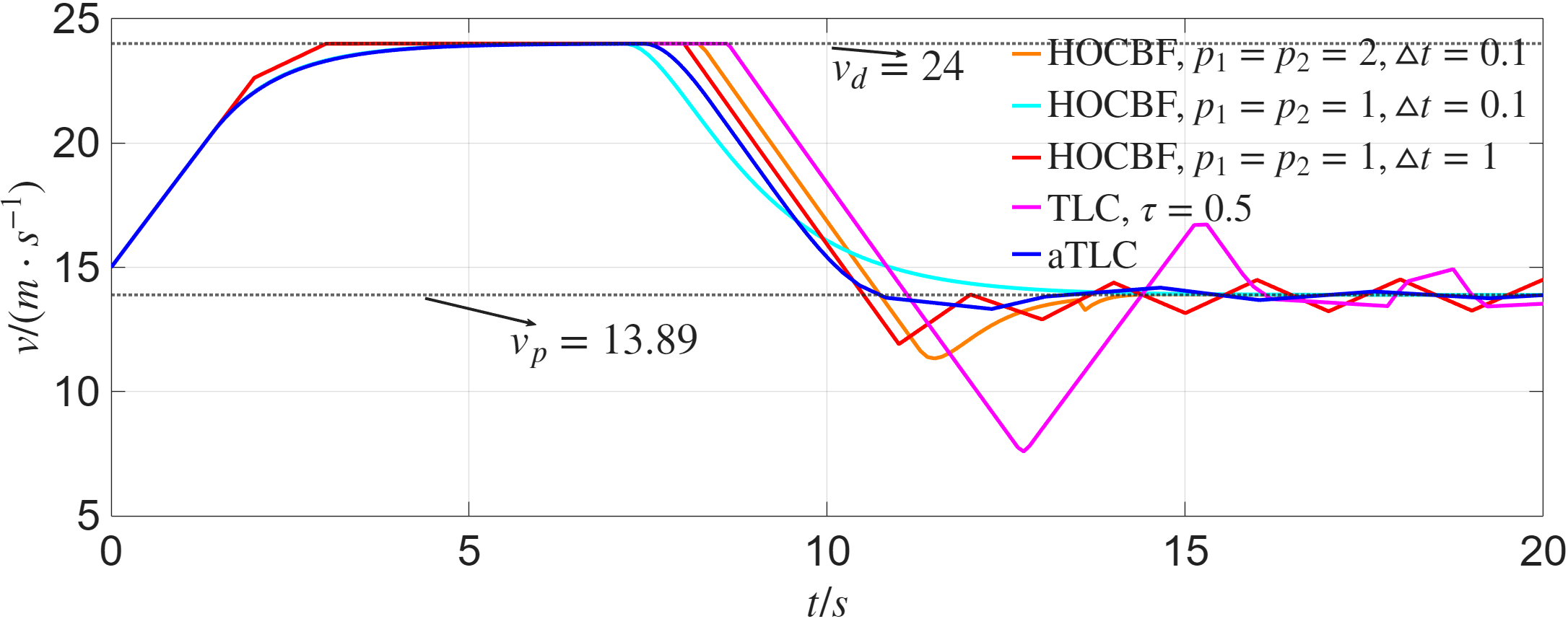}
        \caption{Velocity profiles under different methods when $c_d=0.4$}
        \label{fig:5}
    \end{subfigure}    
    \vspace{0.2cm}    
    \begin{subfigure}{0.99\linewidth}
        \centering
        \includegraphics[width=\linewidth]{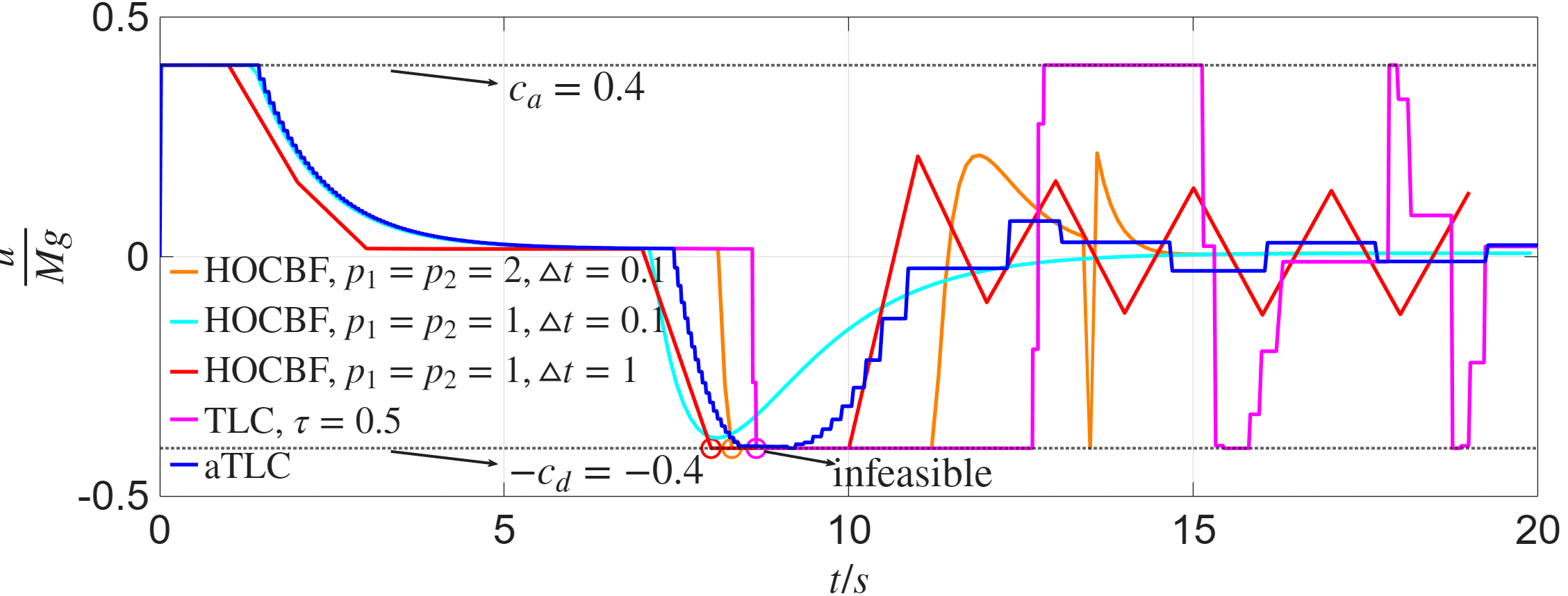}
        \caption{Control input profiles under different methods when $c_d=0.4$}
        \label{fig:6}
    \end{subfigure}    
    \vspace{0.2cm}    
    \begin{subfigure}{0.99\linewidth}
        \centering
        \includegraphics[width=\linewidth]{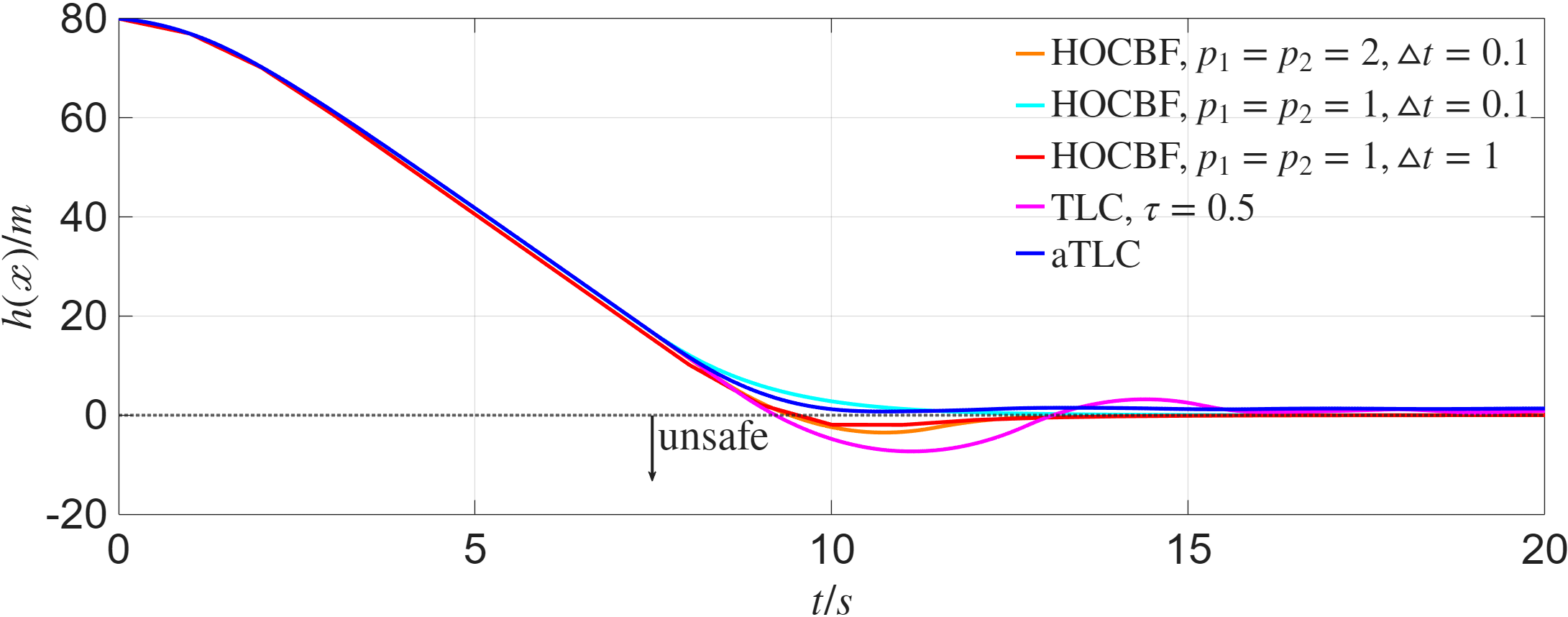}
        \caption{Safety function $h(\boldsymbol{x})$ evolution under different methods when $c_d=0.4$}
        \label{fig:7}
    \end{subfigure}    
    \vspace{0.2cm}    
    \begin{subfigure}{0.99\linewidth}
        \centering
        \includegraphics[width=\linewidth]{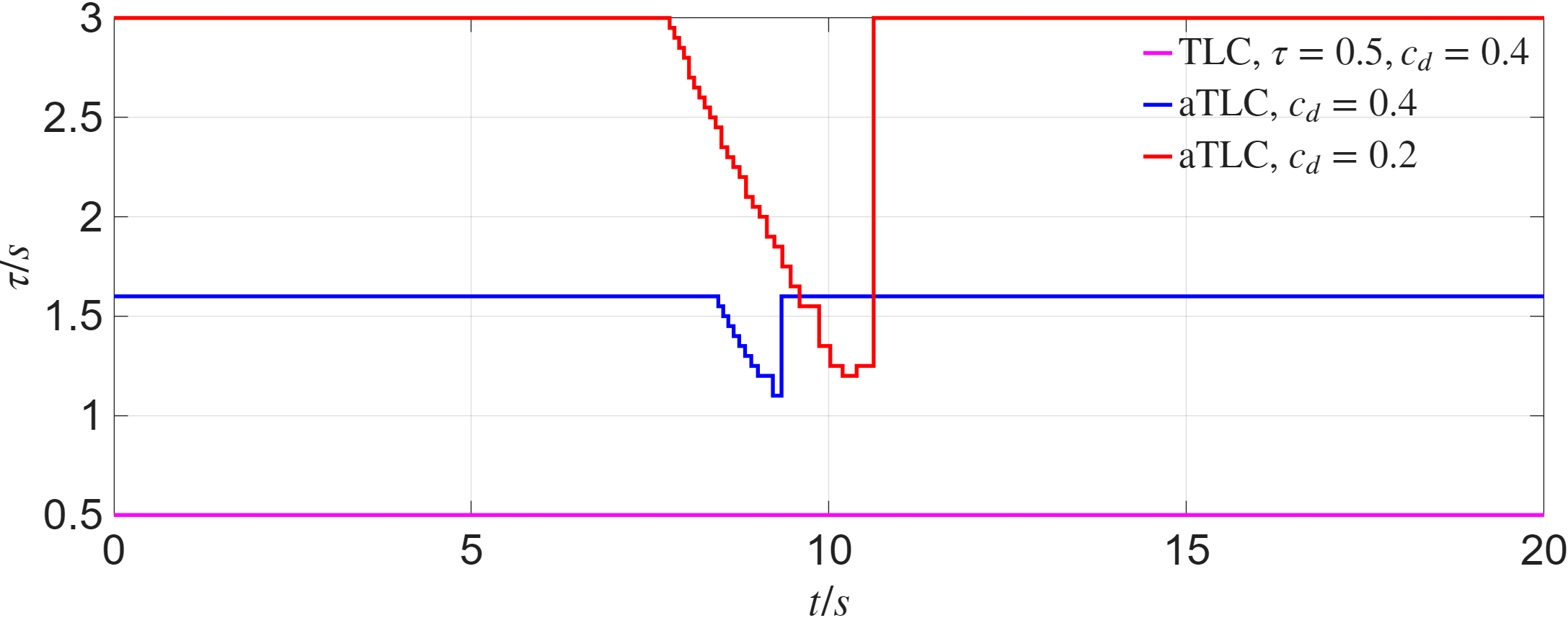}
        \caption{Adaptive $\tau(t)$ in aTLC vs. fixed $\tau$ in TLC}
        \label{fig:8}
    \end{subfigure}   
    \vspace{0.2cm}    
    \begin{subfigure}{0.99\linewidth}
        \centering
        \includegraphics[width=\linewidth]{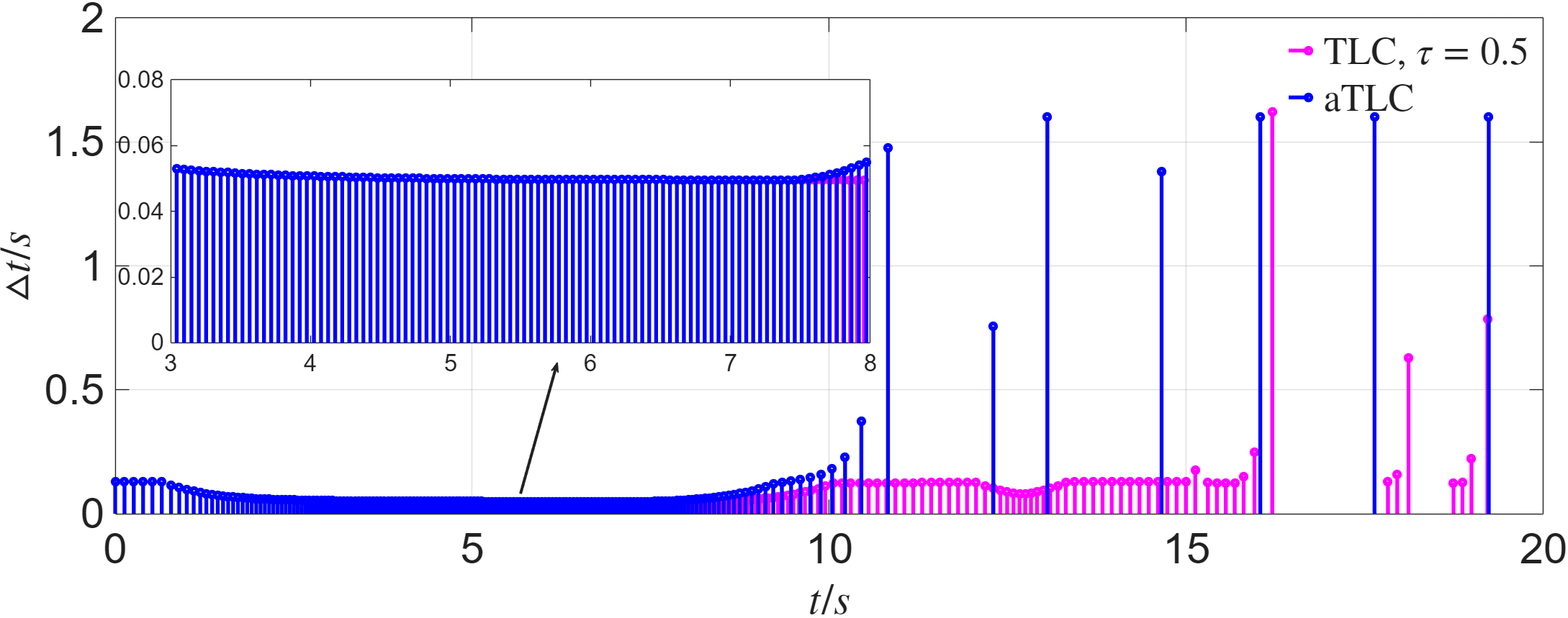}
        \caption{Inter-event time $\Delta t$ under event-triggered implementation for $c_d=0.4$}
        \label{fig:9}
    \end{subfigure}    
    \caption{aTLC improves feasibility and ensures safety compared to TLC, while achieving performance comparable to a well-tuned HOCBF despite requiring only a single parameter.}
    \label{fig:5_vertical}
\end{figure}

\section{Conclusion and Future Work}
\label{sec:conclusion}
This paper proposes an adaptive Taylor--Lagrange Control (aTLC) framework for safety-critical control of nonlinear systems under sampled-data implementations. By treating the time interval as a state-dependent parameter selected online, the proposed method improves feasibility and safety compared to non-adaptive TLC. An event-triggered implementation is developed to mitigate inter-sampling effects, together with an adaptive selection rule that searches a finite candidate set and selects the largest feasible time interval while preserving the QP structure. Simulation results on an adaptive cruise control problem demonstrate that aTLC improves feasibility, maintains safety under limited control bounds, produces smoother control inputs, and reduces triggering events compared to non-adaptive TLC. Future work will focus on extending the proposed framework to systems with model uncertainty, learning-based adaptation of the time interval, and experimental validation on real-world platforms.
\label{sec:Conclusion and Future Work}

\bibliographystyle{IEEEtran}
\balance
\bibliography{references.bib}
\end{document}